\documentclass[10pt, conference, letterpaper]{IEEEtran}
\usepackage[utf8]{inputenc}
\IEEEoverridecommandlockouts
\usepackage{ifpdf}
\usepackage{graphicx}
\DeclareGraphicsExtensions{.eps}
\usepackage{epstopdf}
\usepackage[linesnumbered,ruled, boxed, lined, noend]{algorithm2e}

\usepackage{cite}
\usepackage{amssymb,amsmath,amsthm,amsfonts}
\usepackage{bm}
\usepackage{algorithmic}
\usepackage[linesnumbered,ruled, boxed, lined, noend]{algorithm2e}
\usepackage{textcomp}
\usepackage{xcolor}
\usepackage[linesnumbered,ruled, boxed, lined]{algorithm2e}
\usepackage{url}
\PassOptionsToPackage{hyphens}{url}
\usepackage{bbm}
\usepackage{float}
\usepackage{scalerel}

\def\BibTeX{{\rm B\kern-.05em{\sc i\kern-.025em b}\kern-.08em
    T\kern-.1667em\lower.7ex\hbox{E}\kern-.125emX}}

\newtheorem{proposition}{Proposition}

\allowdisplaybreaks

\DeclareMathOperator*{\argmax}{argmax}

\newcommand{\esp}{\mathbb{E}}

\newcommand{\ttag}{\mathbf{t}}

\begin{document}

\title{Secure Task Offloading and Resource Allocation Design for Multi-Layer Non-Terrestrial Networks
}

\author{

\IEEEauthorblockN{ Alejandro Flores$^*$, Isabella W. G. da Silva$^+$, Vu Nguyen Ha$^*$, Konstantinos Ntontin$^*$, }
\IEEEauthorblockN{ Hien Quoc Ngo$^+$, Michail Matthaiou$^+$ and Symeon Chatzinotas$^*$ }
\IEEEauthorblockA{$^*$Interdisciplinary Centre for Security, Reliability and Trust (SnT), University of Luxembourg, Luxembourg}
\IEEEauthorblockA{$^+$Centre for Wireless Innovation (CWI), Queen’s University Belfast, U.K.}
\IEEEauthorblockA{e-mails:$^*$\{alejandro.flores, vu-nguyen.ha, kostantinos.ntontin, symeon.chatzinotas\}@uni.lu;}
\IEEEauthorblockA{ $^+$\{iwgdasilva01, hien.ngo, m.matthaiou\}@qub.ac.uk} 
}



\newcommand{\SAT}{\mathrm{SAT}}
\newcommand{\LEO}{\mathrm{LEO}}
\newcommand{\UAV}{\mathrm{UAV}}
\newcommand{\HAPS}{\mathrm{HAPS}}
\newcommand{\IoT}{\mathrm{IoT}}
\newcommand{\maxt}{\mathrm{max}}
\newcommand{\mint}{\mathrm{min}}
\newcommand{\Niot}{N_{\mathrm{IoT}}}

\newcommand{\Nuav}{N_{u}^{\mathrm{ULA}}}
\newcommand{\NuavU}{N_{u}^{\mathrm{UPA}}}
\newcommand{\Nhaps}{N_{h}}
\newcommand{\Nleo}{N_{s}}

\maketitle

\begin{abstract}

Remote and resource-constrained Internet-of-Things (IoT) deployments often lack terrestrial connectivity for task offloading, motivating non-terrestrial networks (NTNs) with onboard multiaccess edge computing (MEC) capabilities. Nevertheless, in the presence of malicious actors, authentication needs to be performed to avoid non-authorized nodes from draining the computing resources of the NTN nodes. As a solution, we propose a four-layer MEC-enabled NTN with unmanned aerial vehicles (UAVs) acting as access nodes, a high altitude platform station (HAPS) acting as coordinator and authenticator, and a constellation of low-Earth orbit satellites (LEOSats) acting as remote MEC servers. We consider a tag-based physical-layer authentication (PLA) scheme to authenticate legitimate users, and formulate a joint task offloading decision and resource allocation for the admitted tasks, which is solved via block coordinate descent. Numerical results show that the PLA scheme is efficient and performs better than the benchmark schemes. We also demonstrate that the proposed scheme is robust against malicious attacks even under relaxed false-alarm constraints.


\end{abstract}


\section{Introduction}

Internet-of-Things (IoT)  devices enable continuous data collection and processing across diverse domains, e.g., Industry 4.0 and smart agriculture. 
Nonetheless, these devices are typically constrained by limited computing capabilities, limiting their ability to process data locally in a timely manner. To address this limitation, task offloading to remote servers (using multiaccess edge computing (MEC) platforms)
may be an alternative to achieve low-latency and secure processing. However, even though there are several benefits in employing MEC platforms, the current coverage provided by terrestrial networks (TNs) is limited to approximately 15\% of the Earth’s surface. This leaves arge areas, such as remote regions, disaster zones, or conflict areas, without reliable connectivity~\cite{art:SpaceComs_Ntontin_2025}. 

To overcome this limitation, non-terrestrial networks (NTNs)
have emerged as a promising solution for achieving global coverage by utilizing aerial nodes (e.g., unmanned aerial vehicles (UAVs), high altitude platform station (HAPSs)) and spaceborne assets (e.g., low-Earth orbit satellites (LEOSats)). The dynamic deployment and flexible positioning of these NTN nodes enable seamless service provisioning to IoT devices operating beyond the reach of conventional TN infrastructure~\cite{art:3GPP_Saad_2024}. For instance, 
in~\cite{art:2023_Tang}, the authors explored partial task offloading strategies involving local devices and remote ground servers interconnected via ultra-dense LEOSats backhauls. In the recent work of~\cite{art:preprintAlejandro2026}, the authors considered a HAPS as a central entity, and multiple UAVs for coordination between the IoT devices and LEOSats in a multi-layer MEC-enabled NTN, for efficient service provisioning to IoT deployments over large areas. It was shown that the use of multiple layers can underpin a simplified coordination across wide-spanning IoT deployments. Also, the use of intermediate nodes, such as UAVs and HAPS, for task offloading to MEC-enabled LEOSats can provide higher communication rates to the LEOSats than the direct access from the IoT nodes, resulting in reduced transmission delays.
Nonetheless, the same openness and wide-area accessibility that make MEC-enabled NTNs attractive for remote IoT  applications also expose the offloading interface to malicious nodes, which could attempt to inject task requests and drain the limited onboard computing and backhaul resources, effectively creating a denial-of-service (DoS) condition for legitimate devices. Thus, it is of utmost importance to authenticate the devices for a proper use of the resources of the NTN. However, conventional cryptographic schemes can be computationally infeasible for constrained IoT devices~\cite{art:PLA_IoT_Lee_2020}. To address this challenge, physical-layer authentication (PLA) offers a lightweight alternative by leveraging features of the received signal for rapid verification, making it well-suited for massive IoT access~\cite{8546764}. In the context of NTNs, PLA has been investigated~\cite{10693434}, such as in~\cite{art:Abdrabou_PLANTN_2022}, where the authors exploited the Doppler-shift and the received power as features for a PLA scheme to authenticate the LEOSat constellations. Also, in~\cite{art:Abdrabou_PLANTN_2024}, the optimal detection thresholds to differentiate between legitimate and spoofing satellites were evaluated. 
To the best of our knowledge, existing NTN-PLA studies do not quantify how authentication decisions impact task admission, offloading decisions, and MEC resource allocation in multi-layer NTN computing architectures.

Motivated by the above and acknowledging the benefits of multiple NTN layers for IoT-based task offloading,  we study secure task offloading in a four-layer MEC-enabled NTN where ground IoT devices intend to offload computation tasks in the presence of malicious nodes. The NTN comprises multiple UAVs that forward the received signals in an amplify-and-forward (AF) manner to a HAPS, which serves as a central coordinator and authenticator and subsequently offloads admitted tasks to MEC-enabled LEOSats. Our key contributions are threefold: (i) we integrate a tag-based PLA mechanism into the uplink offloading pipeline and derive expressions for the test statistic distribution, the probability of false alarm (PFA), the probability of detection (PD), and the optimized threshold under a fixed false-alarm constraint; (ii) we couple the PLA-driven admission outcome with a joint task offloading and LEO computing-resource allocation formulation that prevents unauthenticated tasks from consuming MEC resources; and (iii) we propose an iterative block-coordinate descent solution and demonstrate via simulations that the proposed PLA-enabled framework improves feasible task admission and robustness against malicious offloading compared with benchmarks schemes. To the best of our knowledge, this is the first work that co-designs PLA-based authentication and MEC offloading/resource allocation in multi-layer NTNs.


\textit{Notation: } Bold uppercase and lowercase letters denote matrices and vectors, respectively; $(\cdot)^T$ and $(\cdot)^H$ represent the transpose and Hermitian transpose; $\|\cdot\|$ and $|\cdot|$ denote the Euclidean norm and absolute value; $\esp\{\cdot\}$ is the expectation operator, and $\mathbf{a}(\cdot,\cdot)$ denotes the array response of a half-wavelength spaced uniform planar array (UPA); $\mathbb{R}{\cdot}$ extracts the real part, and $\mathbf{x} \sim \mathcal{CN}(\mathbf{0}, \mathbf{C})$ indicates a circularly symmetric complex Gaussian vector with covariance $\mathbf{C}$.
\section{System Model - Joint Authentication, Offloading, and Resource Allocation  Framework}\label{sec:SysModel}
As illustrated in Fig.~\ref{fig:sysmodel}, we consider a heterogeneous network where $I$ remotely-located legitimate resource-constrained IoT devices offload computing tasks to an NTN centrally coordinated by a HAPS, denoted by $a$, with the assistance of $U$ UAVs and in the presence of $M$ malicious nodes. 
The HAPS acts as a CPU that offloads tasks to a set of MEC-enabled LEOSats to provide computing capacity and it manages the synchronization of the UAVs' transmissions for coherent reception.
The malicious nodes attempt to exploit network resources by generating computation-intensive tasks and offloading them to the system, thereby depleting resources available to legitimate devices.
To counteract this, the HAPS first authenticates devices and then allocates computing resources. 
By using a secure channel, the IoT devices and the HAPS agree on a time-varying key-based PLA scheme, which allows the HAPS to verify that the
received message signals are indeed from the intended devices in advance~\cite{10005830}. Specifically, the $i$th device uses its unique secret key $\mathbf{l}_i$, unknown to the malicious nodes, to generate an authentication tag that is superimposed onto the message signal. 
Although the malicious nodes are assumed to be aware of the authentication mechanism, their lack of knowledge about the secret keys prevents successful impersonation attacks. Moreover, since the keys are assumed to be time-varying, if Eve attempts to replay a previous message from one of the IoT devices, it will not be accepted by the HAPS.

Let $\mathcal{I} = \{1,\dots, I\}$, $\mathcal{U}=\{1,\dots,U\}$, $\mathcal{M} = \{1,\dots, M\}$, $\mathcal{U}=\{1,\dots,U\}$ , and $\mathcal{K}=\{1,\dots,K\}$ denote the set of IoT devices, malicious nodes, UAVs and LEOSats, respectively. In this system, the HAPS is assumed to be equipped with $N_{A}^R$ receiving antennas to communicate with the UAVs and $N_{A}^T$ transmit antennas to communicate with the LEOSats. Likewise, the UAVs are equipped with $N_{U}^{R}$ receiving antennas to communicate with the IoT devices and $N_{U}^{T}$ transmitting antennas to communicate with the HAPS. In particular, we assume $N_{U}^{T}=N_{U}^{R}=N_{U}$. Moreover, both the legitimate and malicious nodes are assumed to have a single antenna.


\begin{figure}[t]
    \centering
    \includegraphics[width=\linewidth]{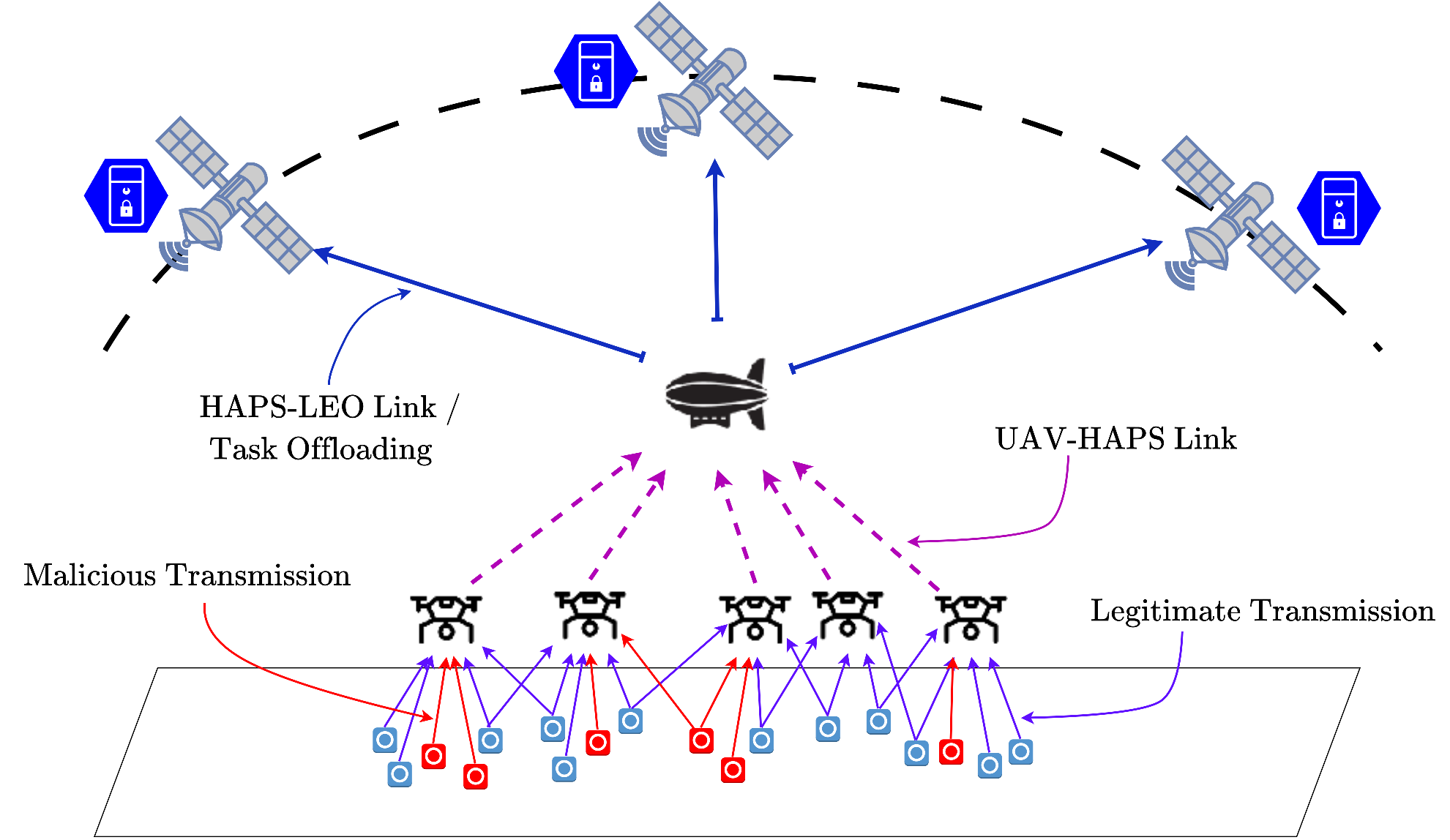}
    \caption{Illustration of the considered system model.}
    \label{fig:sysmodel}
\end{figure}

\subsection{Authentication Framework}\label{sec:authentication}

\subsubsection{Transmission Phase} 

Let $\mathbf{s}_{i} = [s_{i,1}, \dots, s_{i,L}]^T \in \mathbb{C}^{L \times 1}$ denote the message signal from the $i$th IoT device in a particular transmission block with length $L$. An authentication tag $\ttag_{i} = [t_{i,1}, \dots, t_{i,L}] \in \mathbb{C}^{L \times 1}$, which depends on $\mathbf{s}_{i}$ and $\mathbf{l}_i$, is created as
 \begin{align}\label{eq:tag}
     \ttag_{i} = g(\mathbf{s}_{i}, \mathbf{l}_i),     
 \end{align}
 where $g(\cdot)$ is a secure hash function~\cite{bookCrypto}. Here, one assumes that $\esp\{\|\mathbf{s}_i\|^2\}=\esp\{\|\ttag_i\|^2\}=L$ and  $\esp\{s_{i,l}\} = \esp\{t_{i,l}\} = 0, \forall i \in \mathcal{I}, \forall l \in \{1, \dots, L\}$. In addition, the message signals are not correlated to the tags~\cite{4451099}, i.e., $\esp\{\ttag_i^H\mathbf{s}_{i}\}=0$, and since the secret keys are unique to each device, the authentication tags of different devices are independent, i.e., $\esp\{\ttag_i^H\ttag_{i'}\}=0, \forall i' \neq i, i' \in \mathcal{I}$. Then, the tagged signal transmitted by the $i$th IoT device is written as 
 \begin{align}\label{eq:taggedsignal}
     \mathbf{x}_{i} = \rho_{\text{s},i} \mathbf{s}_{i} + \rho_{\text{t},i} \ttag_{i},
 \end{align}
where $\rho_{\text{s},i}$ and $\rho_{\text{t},i}$ are the power allocation coefficients of the message and of the tag signal of device $i$, respectively, which should satisfy $\rho_{\text{t},i} \ll  \rho_{\text{s},i}$ and  $\esp\{\|\mathbf{x}_i\|^2\}\leq L$, and hence, {$\rho_{\text{s},i}^2 + \rho_{\text{t},i}^2 \leq 1$}~\cite{8546764}. Consequently, if $\rho_{\text{s},i}=1$ (then $\rho_{\text{t},i}=0$), it indicates that device $i$ is not transmitting the authentication tag during the transmission block.

Each IoT device transmits its data to the UAVs over orthogonal subchannels with bandwidth $B_i$. Thus, the $N_U \times L$ received signal at UAV $u$ from device $i$ is written as
\begin{align}
     \mathbf{Y}_u = \sqrt{p_i}\mathbf{g}_{i,u}\mathbf{x}_i^H + \bm{\Omega}_u,
\end{align}
where $p_i$ is the maximum normalized transmit power of device $i$, $\bm{\Omega}_u \in \mathcal{C}^{N_u\times L}$ is the noise matrix at UAV $u$, which is assumed to have independent and identically distributed (i.i.d.) $\mathcal{CN}(0, 1)$ entries. 
Also, $\mathbf{g}_{i,u}$$=$$\sqrt{\beta_{i,u}}\mathbf{h}_{i,u}$ is the channel vector from device $i$ to UAV $u$, where $\beta_{i,u}$ is the air-to-ground large-scale fading,
while $\mathbf{h}_{i,u}$ is the small-scale Ricean fading component~\cite{art:Kfactor_Azari_2018}.



Then, the UAVs relay their received signal to the HAPS in a half-duplex AF manner to reduce onboard signal processing and energy consumption at the UAVs. Thus, the aggregated received signal at the HAPS is given by 
\begin{align}
    \mathbf{Y}_a = \scaleobj{.8}{\sum_{u\in\mathcal{U}}}\sqrt{q_{u,i}}\mathbf{G}_{u,a}^H\mathbf{Y}_u  + \bm{\Omega}_a,
\end{align}
where $\bm{\Omega}_a \in \mathcal{C}^{N_{A}^R\times L}$ is the noise matrix at the HAPS, modeled as i.i.d. with $\mathcal{CN}(0, 1)$ entries,  $\mathbf{G}_{u,a}$ is the $ \mathbb{C}^{N_u\times N_{A}^R}$ channel matrix between UAV $u$ and the HAPS, modeled as 

\begin{align}
    \mathbf{G}_{u, a}\!=\!\sqrt{\beta_{u, a}} e^{-j \frac{2\pi}{\lambda} d_{u, a}} \mathbf{a}_a(\vartheta_{u, a}^{\mathrm{AoA}}, \varphi_{u, a}^{\mathrm{AoA}}) \mathbf{a}_u^H(\vartheta_{u, a}^{\mathrm{AoD}}, \varphi_{u, a}^{\mathrm{AoD}}),
\end{align}
where $\beta_{u,a}$, $\lambda$ and $d_{u,a}$ are the large-scale fading, the wavelength,  and the distance between UAV $u$ and the HAPS, respectively. In addition, $\vartheta_{u, a}^{\mathrm{AoA}}$ and $\varphi_{u, a}^{\mathrm{AoA}}$ denotes the azimuth and elevation angles of arrival at the HAPS from UAV $u$, and $\vartheta_{u, a}^{\mathrm{AoD}}$ and $\varphi_{u, a}^{\mathrm{AoD}}$ are the azimuth and elevation angles of departure from UAV $u$ to the HAPS. Finally, $q_{u,i}$ is the amplification factor at UAV $u$ for device $i$, which can be obtained from the following transmit power constraint
\begin{align}
    \esp\{\|\sqrt{q_{u,i}}\mathbf{Y}_u\|^2_F\} &\leq Lp_u, \nonumber\\
    q_{u,i} &= \frac{p_u}{N_U(p_i\beta_{i,u}+1)}, \forall i \in \mathcal{I}, \forall u \in \mathcal{U},
\end{align}
where $p_u$ is UAV $u$'s maximum normalized transmit power.

\subsubsection{Authentication Phase}
Let us assume that the HAPS has perfect knowledge of the weighted effective channel sum between the UAVs and device $i$, $\mathbf{g}_{i,a} \triangleq \sum_{u \in \mathcal{U}}\sqrt{q_{u,i}p_i}\mathbf{G}_{u,a}^H\mathbf{g}_{i,u}$, and employs this knowledge to equalize the received signal and detect $\mathbf{x}_i$. Specifically, the estimated $\hat{\mathbf{x}}_i$ at the HAPS is given by
\begin{align}\label{eq:eq_xi}
    \hat{\mathbf{x}}_i = \frac{\mathbf{g}_{i,a}^H}{\|\mathbf{g}_{i,a}\|^2}\bigg(\scaleobj{.8}{\sum_{u\in\mathcal{U}}}\sqrt{q_{u,i}p_i}\mathbf{G}_{u,a}^H\mathbf{g}_{i,u}\mathbf{x}_i^H + \bar{\bm{\Omega}}_{a}\bigg),
\end{align}
where $\bar{\bm{\Omega}}_{a} \triangleq \sum_{u \in \mathcal{U}}\sqrt{q_{u,i}}\mathbf{G}_{u,a}^H\bm{\Omega}_u + \bm{\Omega}_a$. Given $\rho_{\text{t},i} \ll  \rho_{\text{s},i}$, the presence of the authentication tag in $\hat{\mathbf{x}}_i$ can be ignored, and the estimation of the message signal from device $i$, $\hat{\mathbf{s}}_i$, can be obtained with a tolerable bit error rate (BER)~\cite{7120016}. 

Then, based on $\hat{\mathbf{s}}_i$ and $\mathbf{l}_i$, the HAPS generates an expected authentication tag with the same secure hash function as the one employed by device $i$ as $\Tilde{\ttag}_i = g(\hat{\mathbf{s}}_i, \mathbf{l}_i). $
The HAPS also recovers the received tag from the residual signal as 
\begin{align}\label{eq:eq_ri}
    \mathbf{r}_i = {\left(\hat{\mathbf{x}}_i- \rho_{\text{s},i}\hat{\mathbf{s}}_i\right)}/{\rho_{\text{t},i}}.
\end{align}
To verify whether the estimated tag of device $i$ is present in the received signal, a hypothesis test is conducted as
\begin{align}\label{eq:autdes}
    \left\{\begin{array}{l}
\mathcal{H}_0: \Tilde{\ttag}_i \text{ is not present in } \mathbf{r}_i, \\
\mathcal{H}_1: \Tilde{\ttag}_i \text{ is present in } \mathbf{r}_i,
\end{array}\right. 
\end{align}
where $\mathcal{H}_1$ indicates the presence of the correct tag in the
received signal, whereas $\mathcal{H}_0$ shows that $\mathbf{r}_i$ does not contain the correct tag. Hence, the threshold test for each legitimate device is written as
\begin{align}\label{eq:threshold}
    \lambda_i =  \mathbb{R}\{\Tilde{\ttag}_i^H\mathbf{r}_i\} \stackrel{\mathcal{H}_1}{\underset{\mathcal{H}_0}{\gtrless}} \theta_i, \forall i \in \mathcal{I},
\end{align}
where $\theta_i$ is the decision threshold of device $i$ and $\lambda_i$ is the test statistic, which is constructed by match-filtering $\mathbf{r}_i$ with $\Tilde{\ttag}_i$, that is, $\lambda_i =  \mathbb{R}\{\Tilde{\ttag}_i^H\mathbf{r}_i\}.$
Given the nature of the hash function, if $\mathbf{s}_i$ is recovered with errors, then $\Tilde{\ttag}_i \neq \ttag_i$ with high probability. Thus, we consider that $\hat{\mathbf{s}}_i = \mathbf{s}_i$ and $\Tilde{\ttag}_i = \ttag_i$, and derive the expressions for the test statistics when device $i$ transmits only the message signal ($\lambda_i | \mathcal{H}_0$), and when a tagged signal is received from device $i$ ($\lambda_i | \mathcal{H}_1$) as follows
\begin{align}\label{eq:lambdakh0}
   \lambda_i | \mathcal{H}_0 &=\ttag_i^H\mathbf{r}_i = {\ttag_k^H\mathbf{n}_i}/{\rho_{\text{t},i}},\\
   \lambda_i | \mathcal{H}_1 &= \|\ttag_i\|^2 +{\ttag_k^H\mathbf{n}_i}/{\rho_{\text{t},i}},
\end{align} 
where $ \mathbf{n}_i \triangleq (\mathbf{g}_{i,a}^H\bar{\bm{\Omega}}_{a})/\|\mathbf{g}_{i,a}\|^2$. 

According to the hypotheses, if the HAPS accepts $\mathcal{H}_1$ when $\mathcal{H}_0$ is true, a false alarm will be raised. On the other hand, if $\mathcal{H}_1$ is accepted when $\mathcal{H}_1$ is true, a correct detection has occurred. In the following proposition, we derive the expressions for the PFA and the PD as well as the optimal decision threshold for each device $i$, $\theta^*_i$, to maximize the corresponding PD for a given PFA rate, $\mathsf{P}_{\mathrm{FA},i}$. 
\begin{proposition}\label{prop:cf}
    The closed-form expressions for the PFA, the PD, and the optimal decision threshold of the $i$th IoT device are given, respectively, by
    \begin{align}
        \mathrm{P}_{\mathrm{FA}, i}&=Q\Bigg(\frac{\theta_i}{\sqrt{L \sigma_{n_i}^2/(2\rho_{\text{t},i}^2)}}\Bigg),\label{eq:pfakf}\\
        \mathrm{P}_{\mathrm{D}, i}&=Q\Bigg(\frac{\theta_i-L}{\sqrt{L \sigma_{n_i}^2/(2\rho_{\text{t},i}^2)}}\Bigg),\label{eq:pdkf}\\
    \theta_i^* &= Q^{-1}(\mathsf{P}_{\mathrm{FA},i})\sqrt{L \sigma_{n_i}^2/(2\rho_{\text{t},i}^2)},\label{eq:opttheta}
    \end{align}
    where $Q(\cdot)$ represents the Q-function, and $\sigma_{n_i}^2$ denotes the per-symbol variance of $\mathbf{n}_i$.
\end{proposition}
\begin{IEEEproof}
   The proof follows the methodologies in \cite[Proposition 1]{PLAIsabella}, and the details are omitted due to page constraints. 
\end{IEEEproof}

\subsection{Joint Offloading and Resource Allocation}
The requested computing tasks from device $i$ is embedded onto its message signal $\mathbf{s}_i$, and is denoted by 
\begin{align}
    \pmb{\psi}_i \triangleq [\delta_i, c_i, \tau_i^{\maxt}]^T, \forall i \in \mathcal{I},
\end{align}
where $\delta_i$ is the bit number, $c_{i}$ is the computing density, 
and $\tau_{i}^{\maxt}$ is the maximum tolerable delay. 
Specifically, we consider that the total delay experienced by task $\pmb{\psi}_i$ is given by
\begin{align}\label{eq:ti_0}
    \tau_i = \tau_{i,a} + \scaleobj{.8}{\sum\limits_{k\in\mathcal{K}}}b_{i}^{k}\left(\tau_{a,k} + \tau_k^{\mathrm{prop}}+ \tau_i^k\right),
\end{align}
where $b_i^k$ is a binary variable that indicates if task $i$ is offloaded to LEOSat $k$ ($b_i^k = 1$) or not ($b_i^k = 0$). Moreover, $\tau_{i,a} \triangleq \delta_i / R_{i,a}$ and $\tau_{a,k} \triangleq\sum_{i\in\mathcal{I}} b_{i}^k \delta_i/R_{a,k}$ are, respectively, the transmission delay between device $i$ and the HAPS and between the HAPS and LEOSat $k$, in which $R_{i,a}$ and $R_{a,k}$ are the corresponding communication rates, respectively, i.e., 
\begin{align}
    R_{i,a} &= B_i\log_2\bigg(1+\frac{\rho_{\text{s},i}^2p_i}{\rho_{\text{t},i}^2p_i + \sigma_{n_i}^2}\bigg),\\
    R_{a,k} &= B_a \log_2\bigg( 1 + \rho_{a,k}\frac{\omega_{a,k}^2}{\sigma_k^2} \bigg),
\end{align}
where $B_a$ is the bandwidth of the HAPS-LEO link, $\rho_{a,k}$ is the transmit power of the HAPS for the corresponding link, $\sigma_k^2$ is the noise power at the LEO and $\omega_{a,k}$ is the non-zero singular value of the rank-deficient line-of-sight (LoS) MIMO channel matrix between $a$ and $k$ that reflects the array gain, the path-loss and the atmospheric losses~\cite{rep:itu_t_P.618}. Moreover, $\tau_k^{\mathrm{prop}} \triangleq \frac{1}{c}d_{a,k}$ is the propagation delay between the HAPS and LEOSat $k$, where $c$ is the speed of light, and $d_{a,k}$ is the Euclidean distance between the HAPS and the LEOSat $k$. Finally, $\tau_i^k \triangleq \delta_{i} c_{i}/f_i^k$ is the computing delay of the task from device $i$ at LEOSat $k$, where $f_i^k$ are the resources allocated by LEOSat $k$ to compute the task offloaded by device $i$. We assume a regular offloading pattern in time, for which the optimization is performed once, thus, we do not consider the algorithm runtime delay as an additional source of delay.

\subsection{Problem Formulation}
This work aims to minimize the maximum weighted delay across tasks by selecting the optimized MEC server and computing the optimized share of computing resources for the tasks. Reasonably, the resource allocation should be intended only for legitimate device tasks. Thus, let us define an admission binary variable $\alpha_i$, that indicates whether the task of device $i$ is admitted as legitimate ($\alpha_i=1$) or rejected ($\alpha_i=0$) based on the PLA-based scheme presented in Sec.~\ref{sec:authentication}. Specifically, $\alpha_i=1$ if the test statistic $\lambda_i$ of device $i$ surpasses the optimized detection threshold, $\theta_i^*$. After computing the respective $\alpha_i$ for all $i \in \mathcal{I}\cup\mathcal{M}$, the joint offloading and resource allocation optimization problem can be formulated as
\begin{subequations}\label{eq:main_opt_gen}
\begin{align}\label{eq:main_opt_gen_01}
(\mathcal{P}):\;\; &\min_{\mathbf{F},\mathbf{B}} & &\max\limits_{i\in\mathcal{I}\cup\mathcal{M}} \alpha_i\frac{\tau_i}{\tau_i^{\maxt}}   \\ 
\label{eq:main_opt_gen_02}  &\text{s.t.}&&  \alpha_i\tau_i \leq \tau_{i}^{\maxt},  \;\;  \forall i\in \mathcal{I}\cup\mathcal{M}\\
\label{eq:main_opt_gen_03} & &&  \scaleobj{.9}{\sum\limits_{i\in\mathcal{I}\cup\mathcal{M}}} f_i^k \leq F_k^{\maxt},  \;\;  \forall k\in\mathcal{K},\\
\label{eq:main_opt_gen_04} &      && \scaleobj{.9}{\sum\limits_{k\in\mathcal{K}}}b_i^{k} = \alpha_i,  \;\;\forall i\in\mathcal{I}\cup\mathcal{M},   \\
\label{eq:main_opt_gen_05} &   &&b_{i}^{k} \in \{0,1\},  \;\;  \forall i\in\mathcal{I}\cup\mathcal{M}, \; \forall k\in\mathcal{K},  \\
\label{eq:main_opt_gen_06} &  &&    f_i^k \geq 0,  \;\;   \forall i\in\mathcal{I}\cup\mathcal{M},  k\in\mathcal{K},
\end{align}
\end{subequations}
where $\mathbf{F}\in\mathbb{R}_+^{I\times K}$ and $\mathbf{B}\in\{0,1\}^{I\times K}$ gather the computing resources and offloading decision variables, respectively. Also, $F_k^{\maxt}$ is LEOSat $k$'s maximum computing resource.

\section{Proposed Solution}\label{sec:Algorithm}
To solve the joint authentication, offloading, and resource allocation problem $(\mathcal{P})$ for the admitted users, we start by rewriting it in epigraph form as
\begin{subequations}\label{eq:opt2}
\begin{align}\label{eq:opt2_01}
(\mathcal{P}'):\;\; &\min_{\mathbf{F},\mathbf{B}, \mu} & &\mu\\ 
\label{eq:opt2_08}  &\text{s.t.}&&  \alpha_i\tau_i \leq \mu\tau_{i}^{\maxt},  \;\;  \forall i\in\mathcal{I}\cup\mathcal{M},\\
       \label{eq:opt2_09} &&&0\leq \mu \leq 1,\\
        &&&\eqref{eq:main_opt_gen_03},\eqref{eq:main_opt_gen_04},\eqref{eq:main_opt_gen_05},\eqref{eq:main_opt_gen_06}\nonumber.
\end{align}
\end{subequations}
Given the non-convex nature of $(\mathcal{P}')$, we propose a multi-layer multi-stage optimization framework. 
An iterative block-coordinate descent algorithm is performed over the offloading decision variables and the computing resources to find the solution to $(\mathcal{P}')$, as described in the following. 




\subsection{Offloading Decision}
For the admitted tasks $i\in\mathcal{I}$, and given the set of computing resources $\mathbf{F}$, the following optimization problem is solved to obtain the optimized offloading decisions
\begin{subequations}\label{eq:opt32}
\begin{align}\label{eq:opt3_01}
(\mathcal{P}^o):\;\;\; &\min_{\mathbf{B}, \mu} & &\mu  \\ 
\label{eq:opt3_02}  &\text{s.t.}   && \scaleobj{.9}{\sum\limits_{k\in\mathcal{K}}}b_i^k = 1,  \;\;  \;\;\forall i\in\mathcal{I},   \\
\label{eq:opt3_06} &   &&  b_i^k \in \{0,1\},  \;\;  \forall i\in\mathcal{I}, \; \forall k\in\mathcal{K},  \\
&&&\eqref{eq:opt2_08},\eqref{eq:opt2_09}.\nonumber
\end{align}
\end{subequations}
Problem $(\mathcal{P}^o)$ is not convex due to constraints \eqref{eq:opt2_08} and \eqref{eq:opt3_06}. To deal with these, we first relax constraint \eqref{eq:opt3_06} to allow the decision variables $b_i^k$ to be continuous as $0\leq b_i^k\leq 1$, and solve the problem iteratively with proximity constraints, such as $b_{i}^{k,(n)} - \varrho\leq b_i^k\leq b_{i}^{k,(n)} + \varrho$, where $ \varrho$ is chosen appropriately for convergence and the $^{(n)}$ suffix indicates that the corresponding term is the value from the previous iteration. 

Moreover, a quadratic penalty function $\phi(b_i^k) = (1-b_i^k)b_i^k$ is included in the objective function to drive the solution to binary values. However, given that this function is not convex, we linearize it via a first-order Taylor approximation as $\hat{\phi}(b_i^k) = b_i^k(1-2b_{i}^{k,(n)})$ and define the binary penalty function as $\Phi(\mathbf{B}) = \chi\sum_{k\in\mathcal{K},i\in\mathcal{I}}\hat{\phi}(b_i^k)$, where $\chi$ is increasingly set at each iteration. Then, to handle constraint \eqref{eq:opt2_08}, we start by rewriting the non-convex quadratic term in \eqref{eq:ti_0} as
\begin{align}\label{eq:23}
\frac{b_i^k\sum_{j\neq i}\delta_jb_{j}^{k}}{R_{a,k}} &= \frac{\big(\sum_{j\neq i}\delta_jb_{j}^{k} + b_i^k\big)^2 - \big(\sum_{j\neq i}\delta_jb_{j}^{k} - b_i^k\big)^2}{4R_{a,k}}.
\end{align}
Then, we apply the quadratic transform (QT)~\cite{art:QT_Shen_2018} to each quadratic term in \eqref{eq:23}, such that
\begin{align}
 \hspace{-3mm}   \frac{\big(\sum_{j\neq i}\delta_jb_{j}^{k} \pm b_i^k\big)^2}{R_{a,k}} &\Rightarrow 2z_{\pm,i} \Big| \scaleobj{.8}{\sum_{j\neq i}}\delta_jb_{j}^{k} \pm b_i^k \Big| - z_{\pm,i}^{2} R_{a,k}, \label{eq:24a}\\
    z_{\pm,i} &= {\big| \scaleobj{.8}{\sum_{j\neq i}}\delta_jb_{j}^{k,(n)} \pm b_{i}^{k,(n)} \big|}/{R_{a,k}^{(n)}}.\label{eq:25}
\end{align}
By enforcing $\sum_{j\neq i} \delta_jb_{j}^{k} \geq \delta_{i}$ for all $i\in\mathcal{I}$, we can remove the absolute value from \eqref{eq:24a} and \eqref{eq:25}, and rewrite \eqref{eq:ti_0} as 
\begin{align}
    \tau_i = \zeta_i + \scaleobj{.8}{\sum_{k\in\mathcal{K}}\sum_{j\neq i}}b_{j}^{k}\zeta_{i,j}^{k} + \scaleobj{.8}{\sum_{k\in\mathcal{K}}}b_i^k\zeta_{i}^{k},
\end{align}
where
\begin{align}
    \zeta_i &\triangleq \frac{\delta_i}{R_{i,a}}  - {\scaleobj{.8}{\sum_{k\in\mathcal{K}}\sum_{j\neq i}\delta_jb_{i}^{k,(n)}b_{j}^{k,(n)}}}/{R_{a,k}},\\
    \zeta_{i,j}^k &\triangleq {\delta_jb_{i}^{k,(n)}}/{R_{a,k}},\\
    \zeta_{i}^{k} &\triangleq{\delta_ic_i}/{f_i^k} + {\delta_i}/{R_{a,k}} + {\scaleobj{.8}{\sum_{j\neq i}}\delta_jb_{j}^{k,(n)}}/{R_{a,k}} + 2\tau_k^{\mathrm{prop}}.
\end{align}
Then, let $\pmb{\tau}^{\maxt} = [\tau_1^{\maxt},\dots,\tau_I^{\maxt}]^{T}$, $\mathbf{b}^k=[b_1^k,\dots,b_I^k]^T$, $\mathbf{b} = [\left(\mathbf{b}^{1}\right)^T,\dots,\left(\mathbf{b}^{K}\right)^T,\mu]^T$, $\mathbf{d}=[\delta_1,\dots,\delta_I]^T$, $\mathbf{b}_L = \mathbf{b}^{(n)} + \varrho $, $\mathbf{b}_H = -\mathbf{b}^{(n)} + \varrho $, $\mathbf{B}_I = \mathrm{diag}\left(\mathbf{d}\right) - \mathbf{d}\otimes\mathbf{1}_I^T + \mathbf{I}_I$, $\mathbf{B} = [\mathrm{blkdiag}(\mathbf{B}_I)_K, \mathbf{0}_{IK}]_{IK\times IK+1}$, $\mathbf{C}_L=[\mathrm{blkdiag}(\mathbf{I}_{I})_K, \mathbf{0}_{IK}]_{IK\times IK+1}$, $\mathbf{C}_U=- \mathbf{C}_L$, $\pmb{\zeta} = [\zeta_1,\dots,\zeta_{I}]^T$ and $\mathbf{A} = \left[\mathbf{A}_1, \dots, \mathbf{A}_I, \pmb{\tau}^{\maxt}\right]_{I\times IK+1}$, with
\begin{align}
    \mathbf{A}_k &=\left[\begin{matrix}
        \zeta_{1,1}^{k} & \cdots & \zeta_{1,I}^{k}\\
        \cdots & \cdots & \cdots\\
        \zeta_{I,1}^{k} & \cdots & \zeta_{I,I}^{k}\\
    \end{matrix}\right], \forall k \in \mathcal{K}.
\end{align}
Finally, by defining $\mathbf{\Upsilon} = [\mathbf{A}^T, \mathbf{B}^T,\mathbf{C}_L^T,\mathbf{C}_U^T,\mathbf{C}_L^T,\allowbreak\mathbf{C}_U^T]^T_{5IK + I  \times IK+1}$, $\mathbf{c} = [\pmb{\zeta}^T,\mathbf{0}_{IK}^T, \bar{\mathbf{b}}_L^T, \bar{\mathbf{b}}_U^T, \mathbf{1}_{IK}^T, \mathbf{0}_{IK}^T]^T_{5IK+I \times 1}$ and $\pmb{\upsilon} = [ \chi\left(\mathbf{1}_{IK} - 2\bar{\mathbf{b}}^{(n)}\right)^T, 1 ]^T_{IK+1\times 1}$, where $\bar{\mathbf{b}}^{(n)}$, $\bar{\mathbf{b}}_U$ and $\bar{\mathbf{b}}_L$ correspond to $\mathbf{b}^{(n)}$, $\mathbf{b}_U$ and $\mathbf{b}_L$ without their last element, the problem $(\mathcal{P}^o)$ can be relaxed to a linear program (LP) written as
\begin{align}\label{eq:opt3}
(\mathcal{P}_{\mathrm{LP}}^o):\;\;\; &\min_{\mathbf{b}} \pmb{\upsilon}^T\mathbf{b} \text{  s.t. } \mathbf{\Upsilon}\mathbf{b} \leq \mathbf{c},  \;\;  
\end{align}
which can be readily solved with the convex programming toolbox CVX. If $(\mathcal{P}_L^o)$ is not feasible due to maximum delay constraints of the tasks, we reject the task with the largest pull for resources from the system ($i^*=\argmax_{i\in\mathcal{I}}\{\frac{\delta_ic_i}{\tau_i^{\mathrm{max}}}\}$) until the problem becomes feasible for the remaining tasks.

\subsection{Remote Computing Resource Allocation}
The optimization problem to obtain the optimized computing resources at every LEOSat can be written as
\begin{subequations}\label{eq:opt22}
\begin{align}\label{eq:optf_01}
(\mathcal{P}^{f}):\;\;\; &\min_{\mathbf{F}, \mu} & &\mu  \\ 
\label{eq:optf_02b}  &\text{s.t.}   &&   \scaleobj{.9}{\sum\limits_{k\in\mathcal{K}}}{b_i^k\delta_ic_i}/{f_{i}^{k}}\leq \mu\tau_{i}^{\maxt} - \varepsilon_i,  \;\;  \forall i\in\mathcal{I} \\
&&&\eqref{eq:opt2_09}, \eqref{eq:main_opt_gen_03}, \eqref{eq:main_opt_gen_06}\nonumber,
\end{align}
\end{subequations}
where
$\varepsilon_i = {\delta_i}/{R_{i,a}} + \scaleobj{.9}{\sum_{k\in\mathcal{K}}}b_i^k\big({\scaleobj{.9}{\sum_{j\in\mathcal{I}}}b_{j}^{k}\delta_j}/{R_{a,k}}+2\tau_k^{\mathrm{prop}}\big)$.

This problem is convex and can be solved with dedicated software. Accordingly, problems $(\mathcal{P}^o)$ and $(\mathcal{P}^{f})$ are solved iteratively in a block coordinate descent framework, with increasing binary penalty $\chi$ until convergence. A description of the joint authentication, offloading and resource allocation framework can be seen in \textbf{Algorithm~\ref{alg:algorithm}}.

\SetKwBlock{RepeatE}{repeat}{}
\begin{algorithm}\label{alg:algorithm}
{
    \caption{ \textsc{Joint Authentication, Offloading, and Resource Allocation Optimization}}\label{alg:Solution}
    \ForEach{$j\in\mathcal{I}\cup\mathcal{M}$}{
        Compute $\hat{\mathbf{x}}_j$ as in \eqref{eq:eq_xi} and $\mathbf{r}_j$ as in \eqref{eq:eq_ri}\;
        Compute optimal decision threshold $\theta_j^*$ as in \eqref{eq:opttheta}\;
        Obtain test statistics $\lambda_j$ as in \eqref{eq:threshold}\;
        Admit or reject $j$ as per \eqref{eq:threshold}
     }
    Set $n_{\mathrm{it}} := 1$ and initialize $\mathbf{B}^{(1)}$ and $\mathbf{F}^{(1)}$ for admitted tasks\;
    \Repeat{$n_{\mathrm{it}}\geq N_{\mathrm{it}}$}{
        Compute $\mathbf{\Upsilon}$, $\pmb{\upsilon}$ and $\mathbf{c}$\;
        Solve $\mathcal{P}_{\mathrm{LP}}^o$ and obtain $\mathbf{B}$\;
        \If{$\mathcal{P}_{\mathrm{LP}}^o$ is unfeasible}{
            Reject task $i^*=\argmax_{i}\{\frac{\delta_ic_i}{\tau_i^{\mathrm{max}}}\}$ \;
            Go back to 8\;
        }
        Solve $\mathcal{P}^f$ and obtain $\mathbf{F}$\;
        \textbf{If} $||\mathbf{B}-\mathbf{B}^{(n_{\mathrm{it}})}||_F \leq \epsilon/(IK)$, \textbf{return}\;
        Adjust $\chi$ and $\varrho$\;
        Set $\mathbf{B}^{(n_{\mathrm{it}}+1)} = \mathbf{B}$ and $\mathbf{F}^{(n_{\mathrm{it}}+1)} = \mathbf{F}$\;
        Set $n_{\mathrm{it}} := n_{\mathrm{it}} + 1$\;
    }}\end{algorithm}

\section{Performance Evaluation}\label{sec:Results} 
To evaluate the performance of our PLA-based scheme we consider a system with $I=50$ and $M = 25$, spread over a $10 \times 10$~km$^2$ area following a binomial point process. Both legitimate and malicious nodes generate tasks with a bit number of $10$~Kb and computing density of $200$ CPU cycles per bit. 
The computing capacity at the LEOSats is set as $F_k^{\maxt}=10$~Gcycles/s. As per the 3GPP recommendations~\cite{rep:3gpp_38.101-5}, the carrier frequency for the HAPS-LEO links is chosen in the Ka band, as $28$~GHz, with a bandwidth $B_a=100$~MHz. For the IoT-UAV-HAPS access, we consider NB-IoT-compliant subchannels of bandwidth $B_i = 200$~KHz over the carrier at $2.1$~GHz. The altitude of the UAVs and the HAPS is set as $120$~m and $20$~Km, respectively. Moreover, all UAVs are considered to be uniformly spread in a grid over the area. A set of $K = 3$ LEOSats is considered within a constellation above the HAPS, with the satellites distributed at an Earth-central angular separation of 5 degrees. The constellation is confined within a maximum Earth-central angle of 20 degrees, guaranteeing visibility within the HAPS horizon. Furthermore, unless stated otherwise, $U = 25$, $\rho_{\text{t},i}^2=0.01$, ${\mathsf{P}_{\mathrm{FA}}}=10^{-6}$, $N_u = 64$, and $N_A^R = N_A^T = 256$, $p_i = 20$~dBm $\forall i \in \mathcal{I}$, $\rho_{u}=30$~dBm $\forall u \in \mathcal{U}$, and $\rho_{a,k}=33$~dBm $\forall k \in \mathcal{K}$. Moreover, $\varrho$ is chosen such that $\varrho = 0.01$ for the first iteration, and is uniformly increased to $\varrho = 0.5$ for the last iteration.

\begin{figure}[!t]
    \centering
    \includegraphics[width=85mm, height = 50mm]{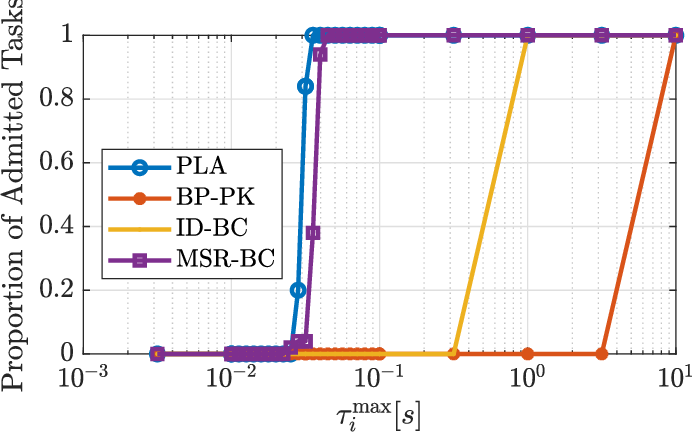}
    \caption{Proportion of feasible tasks for varying $\tau_i^{\maxt}$.}
    \label{fig:res1}
\end{figure}
In Fig.~\ref{fig:res1} we show the proportion of feasible tasks achieved under different authentication schemes as a function of the maximum delay constraint $\tau_i^{\maxt}$. Specially, we consider the following benchmarks: the bilinear pairing-based public-key (BP-PK) scheme from~\cite{art:Ni_Bilinear_2018}, the ID-based blockchain (ID-BC) scheme from~\cite{art:Shen_IBS_2020} and the lightweight MSR-based blockchain (MSR-BC) scheme from~\cite{art:Yang_BC_Auth_2022}. The computation times from these schemes are taken from~\cite{art:Yang_BC_Auth_2022}. Given that the operations for the proposed PLA scheme are performed in the physical layer with minimal signaling, it can be seen that the PLA-based results achieve task feasibility much earlier than all the other considered authentication schemes. This occurs because the time incurred by the authentication schemes entails an extra source of delay that further constraints the tasks to be computed within the maximum time $\tau_i^{\maxt}$. It is worthwhile to mention that although the performance of the MSR-BC scheme approximates to the one observed for the PLA scheme, it requires an additional blockchain infrastructure, which adds significant complexity to the system.

\begin{figure}[!t]
    \centering
    \includegraphics[width=\linewidth]{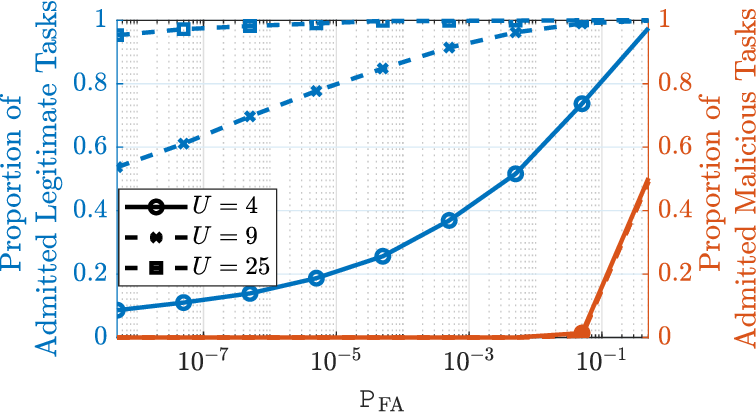}
    \caption{Proportion of legitimate and malicious tasks admitted due to authentication versus the fixed PFA rate, $\mathsf{P}_{\mathrm{FA}}$, for different number of UAVs.}
    \label{fig:res2}
\end{figure}
Figure~\ref{fig:res2} illustrates the proportion of legitimate and malicious tasks admitted during the authentication process as a function of the allowable PFA, $\mathsf{P}_{\mathrm{FA}}$, under varying numbers of UAVs. As expected, increasing $\mathsf{P}_{\mathrm{FA}}$ leads to a higher admission rate for legitimate tasks, due to the corresponding decrease in the optimal detection threshold $\theta_i^*$ for all $i \in \mathcal{I}$. However, this relaxation also slightly raises the acceptance rate of malicious tasks. Notably, even at $\mathsf{P}_{\mathrm{FA}} = 0.05$, the proportion of admitted malicious tasks remains as low as $1$\%. Furthermore, increasing the number of UAVs significantly enhances the detection of legitimate tasks: with $U = 25$, the admission rate for legitimate tasks exceeds $95$\% even under a stringent false alarm constraint.

\begin{figure}[!t]
    \centering
    \includegraphics[width=85mm, height = 50mm]{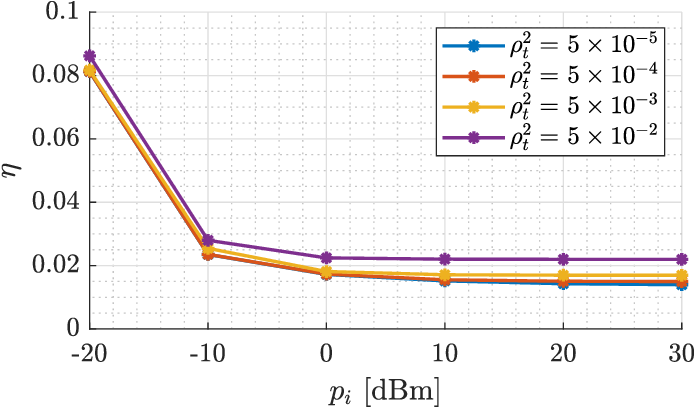}
    \caption{ Maximum relative delay, $\eta$, with varying IoT transmit power, $p_i$. for different tag allocation power ,$\rho_{\text{t},i}^2$.}
    \label{fig:res4}
\end{figure}
In Fig.~\ref{fig:res4}, we illustrate the maximum relative delay $\eta$ versus the transmit power of devices, $p_i$, for different values of the tag allocation power, $\rho_{\text{t},i}^2$. Note that by increasing $p_i$, the IoT device-HAPS communication rates also increase, thereby resulting in a lower offloading delay and consequently lower values of $\eta$. 
However, for $p_i\geq 0$ dBm, there is no significant impact on $\eta$. This occurs because as $p_i$ increases above this value, the contribution of $\tau_{i,a}$ to the total task delay is dominated by the HAPS–LEO transmission latency and computing delays. Furthermore, decreasing $\rho_{\text{t},i}^2$ has minimal effect on $\eta$, thereby enabling a greater portion of power to be allocated to message signal detection.

\section{Conclusions}\label{sec:Conclusions}
In this work, we developed a tag-based PLA framework designed to authenticate IoT devices offloading tasks to a multi-layer NTN for remote processing. We derived the expressions for the PFA and PD for each legitimate user and defined the optimal decision threshold, which was used to admit the tasks for offloading. A joint task offloading and resource allocation optimization problem was formulated for the admitted tasks and solved via block coordinate descent. Our results showed that the PLA scheme is efficient and allows for better task feasibility under stringent task requirements than the compared benchmark schemes. They also showed that increasing the number of UAVs leads to the reliable admission of legitimate tasks even under strict false-alarm requirements. 

\subsection{Future Works}
Potential extensions could include adapting the power allocation to device-specific requirements. Furthermore, it is worth looking into more practical scenarios with imperfect knowledge of the channel state information and imperfect message detection. 



\section*{Acknowledgement}
This work is funded by the Horizon Europe under the Marie Sklodowska Curie actions: ELIXIRION (GA 101120135) and by the U.K. Engineering and Physical Sciences Research Council (EPSRC) grant (EP/X04047X/2) for TITAN Telecoms Hub. The work of H.~Q.~Ngo was supported by the U.K. Research and Innovation Future Leaders Fellowships under Grant MR/X010635/1. The work of M. Matthaiou was supported by the European Research Council (ERC) under the European Union’s Horizon 2020 research and innovation programme (grant agreement No. 101001331). The work of I. W. G. da Silva, H. Q. Ngo and M. Matthaiou was also supported by a research grant from the Department for the Economy
Northern Ireland under the US-Ireland R\&D Partnership Programme. 

\bibliographystyle{IEEEtran}

\bibliography{readme.bib}

\vspace{12pt}

\end{document}